\documentclass[a4paper]{easychair}
\usepackage{clrscode3e}
\usepackage{graphicx}
\usepackage{hyperref}
\usepackage{relsize}
\usepackage{xspace}
\usepackage{amsmath}
\usepackage{amssymb}
\usepackage{multirow}
\usepackage{xcolor}
\usepackage{latexsym}
\usepackage{mathpartir}
\usepackage{tikz}
\usetikzlibrary{arrows}
\usetikzlibrary{positioning,calc}
\usetikzlibrary{fit,shapes}
\usepackage{thmtools, thm-restate}
\usepackage[hang,raggedright]{subfigure}
\usepackage{comment}

\newif\iflongversion\longversiontrue
\iflongversion
\includecomment{longversion}
\excludecomment{shortversion}
\else
\excludecomment{longversion}
\includecomment{shortversion}
\fi

\newcommand\si{SMTInterpol\xspace}

\newcommand\mT{\mathcal{T}}
\newcommand\T{$\mT$\xspace}

\newcommand\store[3]{\ensuremath{#1\langle #2\lhd#3\rangle}\xspace}
\newcommand\select[2]{\ensuremath{#1[#2]}\xspace}
\newcommand\strongeq{\ensuremath{\sim}\xspace}

\newcommand\Weakeq[1]{\ensuremath{\mathrm{WeakEQ}(#1)}}
\newcommand\weakeqi[1][i]{\ensuremath{\approx_{#1}}\xspace}
\newcommand\Weakeqi[2][i]{\ensuremath{\mathrm{WeakEQ}_{#1}(#2)}}
\newcommand\weakpath[1]{\ensuremath{\stackrel{(#1)}{\Leftrightarrow}}\xspace}

\newcommand\weakcongi[1][i]{\ensuremath{\sim_{#1}}}

\newcommand\ta{\ensuremath{\mathcal{T}_A}\xspace}
\newcommand\ti{\ensuremath{\mathcal{T}_I}\xspace}
\newcommand\te{\ensuremath{\mathcal{T}_E}\xspace}
\newcommand\M{\ensuremath{\mathcal{M}}\xspace}
\newcommand\ST{\ensuremath{\mathcal{S}}\xspace}
\newcommand\A{\ensuremath{\mathcal{A}}\xspace}
\newcommand\Stores[1]{\ensuremath{\mathop{\mathrm{Stores}}\left(#1\right)}}

\newcommand\First{\ensuremath{\mathit{Fst}}}
\newcommand\Second{\ensuremath{\mathit{Snd}}}

\newcommand\fresh[1]{\ensuremath{\ulcorner #1\urcorner}}
\newcommand\Cond[1]{\ensuremath{\mathop{\mathrm{Cond}}(#1)}}
\newcommand\Condi[2][i]{\ensuremath{\mathop{\mathrm{Cond}_{#1}}(#2)}}

\renewcommand{\Then}{%
  \textbf{then}\stepcounter{indent} }
\renewcommand{\Else}{%
  \kill\addtocounter{indent}{-1}%
  \liprint\textbf{else}\>\>\stepcounter{indent}}

\newcommand{\type}[1]{\textsc{#1}}

\def\ourtitle{Weakly Equivalent Arrays}

\newtheorem{lemma}{Lemma}

\theoremstyle{definition}
\newtheorem{definition}{Definition}
\newtheorem{example}{Example}

\title{\ourtitle}

\hypersetup{pdfauthor={J{\"u}rgen Christ, Jochen Hoenicke},
  pdftitle={\ourtitle}} 

\author{J{\"u}rgen Christ \and Jochen Hoenicke\thanks{This work is supported
    by the German Research Council (DFG) as part of the Transregional
    Collaborative Research Center ``Automatic Verification and Analysis of
    Complex Systems'' (SFB/TR14 AVACS)}}

\institute{Department of Computer Science, \\ University of
  Freiburg\\ \email{\{christj,hoenicke\}@informatik.uni-freiburg.de}}

\begin{document}
\maketitle

\begin{abstract}
  The (extensional) theory of arrays is widely used to model systems.  Hence,
  efficient decision procedures are needed to model check such systems.
  Current decision procedures for the theory of arrays saturate the
  read-over-write and extensionality axioms originally proposed by McCarthy.
  Various filters are used to limit the number of axiom instantiations while
  preserving completeness.  We present an algorithm that lazily instantiates
  lemmas based on \emph{weak equivalence classes}.  These lemmas are
  easier to interpolate as they only contain existing terms.  We formally
  define weak equivalence and show correctness of the resulting decision
  procedure.
\end{abstract}

\section{Introduction}
Arrays are widely used to model parts of systems.  In software model checking,
for example, the heap of a program can be modelled by an array that represents
the main memory.  A software model checker using such a model can check for
illegal accesses to memory or even memory leaks.  While checking for illegal
accesses can be done using only the axioms proposed by McCarthy, leak checking
typically is done using extensionality.  In this setting, extensionality is
used to ensure that the memory after executing a program does not contain more
allocated memory cells than it contained at the beginning of the program.

The theory of arrays was initially proposed by
McCarthy~\cite{DBLP:conf/ifip/McCarthy62}.  It specifies two operations: (1)
The store operation \store{a}{i}{v} creates a new array that stores at every
index different from $i$ the same value as array $a$ and the value $v$ at
index $i$.  (2) The select operation \select{a}{i} retrieves the value of
array $a$ at position $i$.  The theory is parametric in the index and element
theories.

The store operation only modifies an array at one index.  The values stored at
other indices are not affected by this operation.  Hence, the resulting
array and the array used in the store operation are \emph{weakly equal} 
in the sense that
they differ only at finitely many indices.
Current decision procedures do not fully exploit such dependencies between
arrays.  Instead, they use a series of instantiations of the axiom proposed by
McCarthy to derive weak equivalences.

In this paper we present a new algorithm to decide the quantifier-free
fragment of the theory of arrays.  The decision procedure is based on the
notion of \emph{weak equivalence}, a property that combines equivalence
reasoning with array dependencies.  The new algorithm only produces a few new
terms not present in the input formula during preprocessing.  This is possible
since the decision procedure does not instantiate the axiom proposed by
McCarthy, but axioms derived from them.

\paragraph*{Related Work}
Since the proposal of the theory of arrays by
McCarthy~\cite{DBLP:conf/ifip/McCarthy62} several decision procedures have
been proposed.  We can identify two basic branches: \emph{rewrite-based} and
\emph{instantiation-based} techniques.

Armando et al.~\cite{DBLP:journals/tocl/ArmandoBRS09} used
rewriting techniques to solve the theory of arrays.  They showed how to
construct simplification orderings to achieve completeness.  The benchmarks
used in this paper test specific properties of the array operators like
commutativity of stores if the indices differ.  While these benchmarks require
a lot of instantiations of McCarthy's axioms, they are easy for the decision
procedure presented in this paper since the properties tested by these
benchmarks are properties satisfied by the weak equivalence relation presented
in this paper.

Bruttomesso et al.~\cite{DBLP:journals/corr/abs-1204-2386} present a rewrite
based decision procedure to reason about arrays.  This approach exploits some
key properties of the store operation that are also captured by the weak
equivalence relation described in this paper.  Contrary to our method, the
rewrite based approach is not designed for Nelson--Oppen style theory
combination and thus not easily integratable into an existing SMT solver.
They extended the solver into an interpolating solver for computing
quantifier-free interpolants.  In contrast to our method their solver depends
on the partitioning of the interpolation problem.  We create a SMT proof
without any knowledge of the partitioning and can use proof tree preserving
interpolation~\cite{DBLP:conf/tacas/ChristHN13}, which only requires a
procedure to interpolate the lemmas.

A decision procedure for the theory of arrays based on instantiating
McCarthy's axioms is given by de Moura et
al.~\cite{DBLP:conf/fmcad/MouraB09}.  The decision procedure saturates several
rules that instantiate array axioms under certain conditions.  Several filters
are proposed to minimise the number of instantiations.

Closest to our work is the decision procedure published by Brummayer et
al.~\cite{DBLP:journals/jsat/BrummayerB09}.  Their decision procedure produces
lemmas that can be derived from the axioms for the theory of arrays proposed
by McCarthy.  They consider the theory of arrays with bitvector indices and
prove soundness and completeness of their approach in this setting.  In
contrast to our method, they do not allow free function symbols (i.~e., the
combination of the theory of arrays with the theory of uninterpreted function
symbols) since they only consider a limited form of extensionality where the
extensionality axiom is only instantiated for arrays $a$ and $b$ if the
formula contains the literal $a\neq b$.  We do not have this limitation, but
add some requirements on the index theory that prevent the procedure presented
in this paper from using the theory of bitvectors as index theory.

\section{Notation}\label{sec:notation}
A first order theory consists of a signature $\Sigma$ and a set of models
$\mathbb{M}$.  We assume the equality symbol $=$ with its usual interpretation
is part of any signature.  Every model contains for every sort interpreted by
this model a non-empty domain and a mapping from constant or function symbol
into the corresponding domain.  A theory \T is \emph{stably infinite} if and
only if every satisfiable quantifier-free formula is satisfied in a model of
\T with an infinite universe.

The theory of arrays \ta is parameterised by an index theory \ti and an
element theory \te.  The signature of \ta consists of the two functions
\select{\cdot}{\cdot} and \store{\cdot}{\cdot}{\cdot}.  Every model of the
theory of arrays satisfies the select-over-store-axioms proposed by
McCarthy~\cite{DBLP:conf/ifip/McCarthy62}:
\begin{align}
  \forall a\,i\,v.\ \select{\store{a}{i}{v}}{i} & = v\tag{idx}\label{ax:idx}\\
  \forall a\,i\,j\,v.\ i\neq j\implies\select{\store{a}{i}{v}}{j} & =
  \select{a}{j}\tag{read-over-write}\label{ax:read-over-write}
\end{align}
Additionally we consider the extensional variant of the theory of arrays.
Then, every model has to satisfy the extensionality axiom:
\begin{align}
  \forall a\,b.\ a = b \lor\exists i.\ \select{a}{i} \neq
  \select{b}{i}\tag{ext}\label{ax:ext}
\end{align}

We use $a,b$ to denote array-valued variables, $i,j,k$ to denote index
variables, and $v,w$ to denote element variables.  Additionally we use
subscripts to distinguish different variables.  We use $P$ to denote a
path in a graph.  A path in a graph is interpreted as a sequence of
edges.

In the remainder of this paper, we consider quantifier-free \ta-formulae.
Furthermore we fix the index \ti to a stably infinite theory and the element
theory \te to a theory that contains at least two different
values\footnote{Note that \ta is stably infinite under these conditions.}.

\section{Towards a Nelson--Oppen-based Array Solver}
Multiple theories are usually combined with a variant of the Nelson--Oppen
combination procedure~\cite{DBLP:journals/toplas/NelsonO79}.  The procedure
requires the participating theories to be stably infinite and to only share
the equality symbol $=$.

The procedure first transforms the input such that every literal is
\emph{pure} with respect to the theories.  Let $f(t)$ be a term in the
input.  If $f$ is interpreted by theory $\mT_1$ and $t$ is interpreted
by theory $\mT_2$, then $f(t)$ is not pure.  The first step of the
Nelson--Oppen procedure then generates a fresh variable $v$, rewrites
$f(t)$ into $f(v)$, and adds the definition $v=t$ as a new conjunct to
the formula.  The fresh variable is shared between theories $\mT_1$
and $\mT_2$.  This step is repeated until all terms are pure.  By
abuse of notation, we name the shared variable after its
defining term $t$, e.\,g., we use $\select{a}{i}$ to denote the shared
variable that is defined as $\select{a}{i}$.

Let $V$ be the set of fresh variables introduced in the first step of the
combination procedure.  The second step of the procedure tries to find an
\emph{arrangement} of $V$, i.~e., an equivalence relation between variables in
$V$ such that $\mT_1$ and $\mT_2$ produce partial models that agree with this 
equivalence relation.  Finding such an arrangement is typically done by
propagating equalities or providing case split lemmas.  In the following, we
call this arrangement strong equivalence to distinguish it from weak
equivalence defined in the next section.  We write $a \strongeq b$ to
denote that $a$ and $b$ are strongly equivalent, i.e., that in the current
arrangement the shared variables $a$ and $b$ are equal.

For the theory of arrays, we consider every term of the form
$\store{\cdot}{\cdot}{\cdot}$ or $\select{\cdot}{\cdot}$ as being
interpreted by the array theory.  
We consider all array terms, store, and select terms to be shared and 
thus they have to occur in the arrangement. 
Furthermore, every index term that
appears in a store or select is considered shared between the array
theory and the index theory.  Then the goal is to find a suitable
arrangement to these shared terms such that all theories agree on this
arrangement.

For an array solver to be used in Nelson--Oppen combination we have to
propagate equalities between shared array terms and shared select terms.
Furthermore, the other theories have to propagate equalities between terms
used as index in a select or store.  In the remainder of this paper we will
first show how to propagate equalities between select terms and afterwards
deal with extensionality to propagate equalities between array-valued terms.

\section{Weak Equivalences over Arrays}
The theory of arrays has two constructors for arrays: array variables, and
store terms $\store{\cdot}{\cdot}{\cdot}$.  Assuming quantifier-free input, we can only
constrain the values of a finite number of indices.  These constraints can
either be explicity like $\select{a}{i}=v$, or implicit like $\store{a}{i}{v}$
where axiom (\ref{ax:idx}) produces the corresponding
$\select{\store{a}{i}{v}}{i}=v$.  Hence, for quantifier-free input, arrays
that are connected via a sequence of $\store{\cdot}{\cdot}{\cdot}$ can only
differ in finitely many positions.  We call such arrays \emph{weakly
  equivalent}.  In this section we formally define weak equality and show how
to exploit this to produce a decision procedure for the (extensional) theory of
arrays.

Let \ST be the set of all terms of the form \store{\cdot}{\cdot}{\cdot} in the
input formula and \A be the set of all array-valued terms that are not in \ST.
Since $\store{a}{i}{v}$ modifies $a$ only at index $i$, these two arrays are
guaranteed to be equal on all indices except on index $i$.  We generalise this
observation to chains of the form $\store{\store{\ldots}{j}{w}}{i}{v}$ to
extract a set of indices for which two arrays might store different values.

\begin{definition}[weak equivalence]\label{def:weakeq}
A \emph{weak equivalence graph} $G^W$ contains vertices
$\ST\cup\A$ and undirected edges defined as follows:
\begin{enumerate}
\item $a\leftrightarrow b$ if $a\strongeq b$, and
\item $a\stackrel{i}{\leftrightarrow} b$ if $a$ has form
  \store{b}{i}{\cdot}.
\end{enumerate}
We write $a\weakpath{P} b$ if there exists a path $P$ between nodes $a$ and
$b$ in $G^W$.  In this case, we call $a$ and $b$ \emph{weakly equal}.  The weak
equivalence class containing all elements that are weakly equal to $a$ is
defined as $\Weakeq{a} := \{b\ |\ \exists P.\ a\weakpath{P} b\}$.
\end{definition}
For a path $P$ we define $\Stores{P}$ as the set of all indices corresponding
to edges of the form $\stackrel{\cdot}{\leftrightarrow}$, i.~e.,
$\Stores{P}:=\{i\ |\ \exists a\,b.\ a\stackrel{i}{\leftrightarrow} b\in P\}$.

\begin{example}\label{ex:weakeq}
  Consider the formula $a=\store{b}{j}{v}\land b=\store{c}{i}{w}\land d=e\land
  \select{c}{i}=w$.  The weak equivalence graph for this example is shown in
  Figure~\ref{fig:weakeq}.  Note that the last conjunct is not important for
  the construction of the weak equivalence graph.
\begin{figure}[htpb]
  \centering
  \begin{tikzpicture}[<->,node distance=1.5cm]
    \node (d) at (0,0) {$d$};
    \node[right of=d] (e) {$e$};
    \node[above of=d,node distance=.5cm] (b) {$b$};
    \node[right of=b] (storec) {$\store{c}{i}{w}$};
    \node[right of=storec] (c) {$c$};
    \node[left of=b] (storeb) {$\store{b}{j}{v}$};
    \node[left of=storeb] (a) {$a$};
    \draw (d) -- (e);
    \draw (a) -- (storeb);
    \draw (storeb) --node[above]{$j$} (b);
    \draw (b) -- (storec);
    \draw (storec) --node[above]{$i$} (c);
  \end{tikzpicture}
  \caption{\label{fig:weakeq}Weak Equivalence Graph for
    Example~\ref{ex:weakeq}}
\end{figure}

  We get two different weak equivalence classes.  The first one contains the
  nodes $a$, $\store{b}{j}{v}$, $b$, $\store{c}{i}{w}$, and $c$.  The second
  contains $d$ and $e$.  Note that $d$ and $e$ are actually strongly equivalent.
  Thus, they store the same value at every position.  Let $P$ denote the path
  from $a$ to $c$ in the weak equivalence graph.  Then, $\Stores{P}=\{i,j\}$.
  Thus, arrays $a$ and $c$ can only differ in at most the values stored at the
  indices $i$ and $j$.
\end{example}

If we want to know if $\select{a}{i}$ and $\select{b}{i}$ should be equal, we
check if $a\weakpath{P}b$ for a path $P$ such that $i\not\in\Stores{P}$.  If
this is the case, $P$ witnesses the equivalence between the select terms.

\begin{definition}[weak equivalence modulo $i$]\label{def:weakeqi}
  Two arrays $a$ and $b$ are \emph{weakly equivalent modulo~$i$} if and only if
  they are weakly equivalent and connected by a path that does not contain an
  edge of the form $\stackrel{j}{\leftrightarrow}$ where $j\strongeq i$.  We
  denote weak equivalence modulo~$i$ by $a\weakeqi b$ and define it as
  $a\weakeqi b := \exists P.\ a\weakpath{P} b\land \forall
  j\in\Stores{P}.\ j\not\strongeq i$.
\end{definition}

Using this definition we can propagate equalities between shared selects if
the arrays are weakly equivalent modulo the index of the select.

\begin{lemma}[read-over-weakeq]\label{lem:read-over-weakeq}
  Let $\strongeq$ be an arrangement satisfying the array axioms.
  Let $\select{a}{i}$ and $\select{b}{j}$ be two selects such that $i\strongeq
  j$ and $a\weakeqi b$.  Then, $\select{a}{i} \strongeq \select{b}{j}$ holds.
\end{lemma}
\begin{proof}
  We induct over the length of the path $P$ witnessing $a\weakeqi b$.

  \paragraph{Base case.}
  In this case, $a$ and $b$ are the same term. Hence
  $\select{a}{i}\strongeq\select{b}{j}$ holds by congruence.

  \paragraph{Step case.}
  Let the step from $c$ to $b$ be the last step of path $P$.
  By induction hypothesis we know that 
  $\select{a}{i}\strongeq\select{c}{j}$ holds.

  If the edge between $c$ and $b$ is due to a strong equivalence (i.~e.,
  $c\strongeq b$), then $\select{c}{j}\strongeq\select{b}{j}$ follows from
  congruence.

  If the edge between $c$ and $b$ is of the form
  $c\stackrel{k}{\leftrightarrow} b$, then either $c$ is $\store{b}{k}{\cdot}$
  or $b$ is $\store{c}{k}{\cdot}$.  In both cases, we get the lemma
  $j=k\lor\select{c}{j}=\select{b}{j}$ from axiom (\ref{ax:read-over-write}).
  Since $j\strongeq i$ and $i\not\strongeq k$, we get $j\not\strongeq k$ and
  thus $\select{c}{j}\strongeq\select{b}{j}$.  We conclude
  $\select{a}{i}\strongeq\select{b}{j}$ by transitivity.
\end{proof}

This lemma allows us to propagate equalities between shared selects.  Note
that it depends upon disequalities between index terms needed to ensure
$a\weakeqi b$.

If two arrays are weak equivalent modulo~$i$ they store the same value at the
index $i$.  The reverse is not necessarily true.  Therefore, we define a
weaker relation weak congruence modulo~$i$. 
\begin{definition}[weak congruence modulo~$i$]\label{def:weakcongi}
  Arrays $a$ and $b$ are \emph{weak congruent modulo~$i$} if and only if they
  are guaranteed to store the same value at index $i$.  We denote weak
  congruence modulo~$i$ by $\weakcongi$ and define $a\weakcongi b := a\weakeqi
  b\lor\exists a'\,b'\,j\,k.\ a\weakeqi
  a'\land i\strongeq j\land
  \select{a'}{j}\strongeq\select{b'}{k}\land k \strongeq i
  \land b'\weakeqi b$.
\end{definition}

We use weak congruences to decide extensionality.  Intuitively, if for all
indices $i$ the weak congruence modulo~$i$ $a\weakcongi b$ holds, then $a=b$
should be propagated.  But this na\"{\i}ve approach requires checking every 
index occurring in the formula.  To minimise the number of indices we need to
consider, we exploit the weak equivalence graph.
\begin{lemma}[weakeq-ext]\label{lem:weakeq-ext}
  Let $\strongeq$ be an arrangement satisfying the array axioms.
  Let $a$ and $b$ be two arrays such that $a\weakpath{P}b$ holds.  If for all
  indices $i\in\Stores{P}$ we have $a\weakcongi b$, then $a\strongeq b$ holds.
\end{lemma}

\begin{proof}
  Follows from Lemma~\ref{lem:read-over-weakeq},
  Definition~\ref{def:weakcongi} and (\ref{ax:ext}).
\end{proof}


\section{A Decision Procedure Based on Weak Equivalences}
Our decision procedure is based on weak equivalences and the
Nelson--Oppen combination scheme.  It propagates equalities between
terms shared by multiple theories.  We limit the propagation to shared array
terms and array select terms.

The \ta-formulae are preprocessed as follows.  For every
\store{a}{i}{v} contained in the input, we (1) instantiate
the axiom (\ref{ax:idx}) and (2) add \select{a}{i} to the set of terms
contained in the input\footnote{This can be achieved by adding the
equality $\select{a}{i} = \select{a}{i}$.}.  Thus, the preprocessing
step adds at most two select operations for every store.

We propagate new equalities from weak equivalence relations and weak
congruence relations based on lemmas
\ref{lem:read-over-weakeq}~and~\ref{lem:weakeq-ext}.  These relations
depend on the arrangement $\strongeq$, which represents logical
equality ($=$).  We now define a function $\Cond\cdot$ that computes a
condition (a conjunction of equalities and inequalities) under which
a weak equivalence or weak congruence holds. To denote
the condition for a path that does not contain an edge labelled with
index $i$ we use $\Condi{\cdot}$.  For an edge in the weak equivalence
graph that represents an equality, the condition reflects this equality.
For an edge that comes from a $\store{\cdot}{j}{\cdot}$, no condition
is needed.  However, $\Condi{\cdot}$ should ensure that $i$ does not
occur on the path, so $i \neq j$ needs to hold.
\begin{align*}
  \Cond{a\leftrightarrow b} &:= a=b &
  \Condi{a\leftrightarrow b} &:= a=b\\
  \Cond{a\stackrel{j}{\leftrightarrow} b} &:= \mathrm{true} &
  \Condi{a\stackrel{j}{\leftrightarrow} b} &:= i\neq j
\end{align*}  
We can extend these definitions to paths by conjoining the conditions for all
edges on that path.  Then, we can compute $\Cond{a\weakeqi b}$ using the path
that witnesses $a\weakeqi b$.
\[
\Cond{a\weakeqi b} := \Condi{P}\text{ where }a\weakpath{P}b\land\forall
j\in\Stores{P}.\ i\not\strongeq j
\]
Finally, to define $\Cond{a\weakcongi b}$, we use the definition of
$\weakcongi$.
\[
\Cond{a\weakcongi b} :=\begin{cases}
\Cond{a\weakeqi b}&\text{if }a\weakeqi b\\
\begin{array}{l}
\Cond{a\weakeqi a'}\land i = j \land\select{a'}{j}=\select{b'}{k}\\
\quad{}\land k = i\land\Cond{b'\weakeqi b}
\end{array}
&\text{if }
\begin{array}{l}
a\weakeqi a'\land i\strongeq j \land \select{a'}{j}\strongeq\select{b'}{k}\\
\quad{}\land k \strongeq i \land b'\weakeqi b
\end{array}
\end{cases}
\]

\begin{example}\label{ex:condi}
  Consider again the formula $a=\store{b}{j}{v}\land b=\store{c}{i}{w}\land
  d=e\land \select{c}{i}=w$ from Example~\ref{ex:weakeq} whose weak
  equivalence graph is shown in Figure~\ref{fig:weakeq}.  Assume $i\not\strongeq j$.
  Then we have $a\weakeqi\store{c}{i}{w}$ since no edge contains a label that
  is equivalent to $i$.  We get $\Cond{a\weakeqi\store{c}{i}{w}}\equiv
  a=\store{b}{j}{v}\land i\neq j\land b=\store{c}{i}{w}$.

  From Axiom~(\ref{ax:idx}) we get $\select{\store{c}{i}{w}}{i}=w$.  With
  $\select{c}{i}=w$ we conclude $a\weakcongi c$ since
  $a\weakeqi\store{c}{i}{w}$ and $\select{\store{c}{i}{w}}{i} =
  \select{c}{i}$.  We have $\Cond{a\weakcongi c} \equiv
  \Cond{a\weakeqi\store{c}{i}{w}}\land\select{\store{c}{i}{w}}{i} =
  \select{c}{i}$.
\end{example}

To decide the theory of arrays we define two rules to generate
instances of array lemmas.  We present the rules as inference rules.
The rule is applicable if the current arrangement $\strongeq$ on
the shared variables $V$
satisfies the conditions above the line.  The rule then generates a
new (valid) lemma that can propagate an equality under the current
arrangement.

The first rule is based on Lemma~\ref{lem:read-over-weakeq}.  Two select
terms are equivalent if the indices of the selects are congruent and the
arrays are weakly equivalent modulo that index.  We only create this lemma
if the select terms existed in the formula.  Note that we create for
select terms in the formula a shared variable with the same name in $V$.  
\[
\inferrule*{a\weakeqi b\\i\strongeq j\\\select{a}{i},\select{b}{j}\in V}
           {i\neq j\lor\lnot\Cond{a\weakeqi b}\lor
             \select{a}{i}=\select{b}{j}}
           \tag{read-over-weakeq}\label{rule:read-over-weakeq}
\]

The next rule is based on Lemma~\ref{lem:weakeq-ext} and used to
propagate an equality between two extensionally equal array terms.  Two
arrays $a$ and $b$ have to be equal if there is a path $P$ such that
$a\weakpath{P}b$ and for all $i\in\Stores{P}$, $a\weakcongi b$ holds.
\[
\inferrule*{a\weakpath{P}b\\ \forall i\in\Stores{P}.\ a\weakcongi b\\a,b\in V}
           {\lnot\Cond{P}\lor\bigvee_{i\in\Stores{P}}\lnot\Cond{a\weakcongi b}
             \lor a=b}
           \tag{weakeq-ext}\label{rule:weakeq-ext}
\]

The resulting decision procedure is sound and complete for the existential
theory of arrays assuming sound and complete decision procedures for the index
and element theories.

\begin{lemma}[soundness]\label{lem:soundness}
  Rules (\ref{rule:read-over-weakeq}) and (\ref{rule:weakeq-ext}) are sound.
\end{lemma}
\begin{proof}
  Soundness of the rules follows directly from the lemma with the
  corresponding name.
\end{proof}

\begin{restatable}[completeness]{lemma}{completeness}\label{lem:completeness}
  The rules (\ref{rule:read-over-weakeq}) and (\ref{rule:weakeq-ext}) are
  complete.
\end{restatable}

\begin{longversion}
\begin{proof}
  Assume all rules are saturated.  Let \M be the model generated by
  the theories different from the array theory.  In this model, arrays
  are considered uninterpreted and only subject to congruence.  Note
  that $v_1\strongeq v_2$ if and only if $\M(v_1)=\M(v_2)$ is
  guaranteed in this model.  We create a new model $\M_A$ that extends
  \M by the interpretation of the array terms in the formula.

  First, for every array type $\sigma\Rightarrow\tau$, we define its domain as
  the set of all functions from $\sigma$ to $\tau$.  The interpretation of
  $\select{\cdot}{\cdot}$ is function application and the interpretation of
  $\store{\cdot}{\cdot}{\cdot}$ is function update.  This definition trivially
  satisfies the array axioms (\ref{ax:idx}) and (\ref{ax:read-over-write}).
  For every other function symbol $f$ and every constant $v$ that is not of array type
  we define $\M_A(f) = \M(f)$ and $\M_A(v) = \M(v)$.
  Next we define the interpretation for all constants of type array in $\M_A$,
  such that $\M_A$ satisfies the input formula.

  Let $\prec$ be a partial order on types such that for every type
  $\sigma\Rightarrow\tau$ we have $\sigma\prec\sigma\Rightarrow\tau$ and
  $\tau\prec\sigma\Rightarrow\tau$.  We define the interpretation of constants
  according to $\prec$.  Thus, when defining the interpretation for a constant
  of type $\sigma\Rightarrow\tau$ we assume all constants of type $\sigma$
  resp.\ $\tau$ are already defined.

  For each sort $\tau$ we assume two different values $\First_\tau$ and
  $\Second_\tau$.  Furthermore we assume every sort $\sigma$ that is used as
  index sort for an array sort contains an infinite supply of fresh domain
  elements denoted by $\fresh{\cdot}$.  Then, for an array constant $a$ of sort
  $\sigma\Rightarrow\tau$, we define
  \[
  \M_A(a)(\mbox{\j}):=\begin{cases}
  \M(\select{b}{i})&\text{if }\select{b}{i}\text{ occurs in input, }\M(i)=\mbox{\j}\text{ and } a\weakeqi b\\
  \Second_\tau&\text{if }\mbox{\j}=\fresh{\Weakeq{a}}\\
  \First_\tau&\text{otherwise}
  \end{cases}
  \]
  The first case is well defined.  Given two
  select terms $\select{b_1}{i_1}$ and $\select{b_2}{i_2}$ with
  $a\weakeqi[i_1] b_1$, $a\weakeqi[i_2] b_2$ and $\M(i_1) = \mbox{\j} = \M(i_2)$ we
  need to show $\M(\select{b_1}{i_1})=\M(\select{b_2}{i_2})$.  
  Since $\M(i_1)= \M(i_2)$, we have $i_1 \sim i_2$ and
  $b_1\weakeqi[i_1] b_2$ (since $\weakeqi[i_1]$ and $\weakeqi[i_2]$ is
  the same relation).  The rule (\ref{rule:read-over-weakeq})
  generated the lemma 
  $i_1\neq i_2\lor\lnot\Cond{b_1\weakeqi[i_1] b_2}\lor
  \select{b_1}{i_1}=\select{b_2}{i_2}$.  
  \M guarantees that 
  $i_1=i_2$ and $\Cond{b_1\weakeqi[i_1] b_2}$ hold.
  Thus, $\M(\select{b_1}{i_1})=\M(\select{b_2}{i_2})$ has to hold in
  order to satisfy this lemma.

  We have to show that the input formula is satisfied by $\M_A$ if it
  is satisfied by $\M$.
  We assume that the array operations in the input formula were
  flattened by introducing fresh variables, i.\,e., that
  $\select\cdot\cdot$ and $\store\cdot\cdot\cdot$ occur only in
  definitions $v=\select{a}{i}$ and $b = \store{a}{i}{v}$.  By
  definition of $\M_A$ it already satisfies the parts of the input
  formulae that do not involve arrays.  It remains to show that
  \begin{enumerate}
  \item for all definition $v=\select ai$ 
    $\M_A(v) = \M_A(\select ai)$ holds.
  \item for all definition $b=\store a i v$ 
    $\M_A(b) = \M_A(\store aiv)$ holds.
  \item $\M_A(a) = \M_A(b)$ if and only if $\M(a) = \M(b)$ for array constants
    $a$ and $b$, and
  \end{enumerate}

  \begin{enumerate}
  \item For a definition $v=\select{a}{i}$ the model $\M$ already guarantees
  $\M(v) = \M(\select{a}{i})$.  The definition of $\M_A$ gives us
  $\M_A(v) = \M(v) = \M(\select ai) = \M_A(a)(\M(i)) = \M_A(\select ai)$ 
  as required.
  
  \item For $b=\store{a}{i}{v}$ we need to show
  $\M_A(b)(\mbox{\j}) = \M_A(a)(\mbox{\j})$ for $\mbox{\j}\neq \M(i)$ and $\M_A(b)(\M(i)) =
  \M(v)$.
  Our preprocessing step adds the equality $\select{\store{a}{i}{v}}{i}=v$
  to the input.  Thus, $\M_A(b)(\M(i)) = \M(\select{\store{a}{i}{v}}{i})
  = \M(v)$. For $\mbox{\j}\neq \M(i)$, we can derive
  $\M_A(b)(\mbox{\j}) = \M_A(a)(\mbox{\j})$ from $\Weakeq{b} = \Weakeq{a}$ and the fact
  that $a\weakeqi[j] b$ for every $j\not\strongeq i$.

  \item We show that our extended model satisfies equalities on arrays.  If
  $a\strongeq b$ holds for two arrays $a$ and $b$, then by construction
  $\M_A(a)=\M_A(b)$.  If $a\not\strongeq b$ holds for two arrays $a$ and $b$,
  then we distinguish two cases.  If $a$ and $b$ are not connected in the weak
  equivalence graph, they differ at the indices $\fresh{\Weakeq{a}}$ and
  $\fresh{\Weakeq{b}}$.  Hence, $\M_A(a)\neq\M_A(b)$ holds.  Otherwise, since
  rule (\ref{rule:weakeq-ext}) is saturated and $a\weakpath{P} b$ for some 
  path $P$ but $a\neq b$, there exists an $i\in\Stores{P}$ such that
  $a\not\weakcongi b$ holds.
  Let $a'$ be the first array on the path $P$ that involves an edge
  $\stackrel{i}\leftrightarrow$ and $b'$ the last such array.
  In the preprocessing step we added $\select {a'}{i}$ and $\select {b'}{i}$.
  By the choice of $a'$ and $b'$ we have $a\weakeqi a'$ and $b' \weakeqi b$.
  Since $a\not\weakcongi b$ holds, we have $\select{a'}{i} \not
  \strongeq\select{b'}{i}$. By definition $\M_A(a)(\M(i)) = \M(\select{a'}{i})$
  and $\M_A(b)(\M(i)) = \M(\select{b'}{i})$, so $\M_A(a)$ and $\M_A(b)$ are
  different arrays as desired.
  \end{enumerate}
\end{proof}
\end{longversion}
\begin{shortversion}
\noindent The proof of this lemma can be found in the extended version of this
paper~\cite{longversion}.
\end{shortversion}

\section{Restricting Instantiations}
The preprocessor is the only component of our decision procedure that produces
new select terms and thus might trigger new lemmas.  These lemmas only
generate new (dis-)equality literals between existing terms.  Thus, reducing
the number of select terms might reduce the number of lemmas generated by our
decision procedure and speed up the procedure.

If the element theory is stably infinite we can omit the preprocessor step
that adds for every $\store{a}{i}{v}$ the select $\select{a}{i}$.  Instead, we
simply assume $\select{a}{i}$ to be different than any other $\select{b}{i}$
unless $a\strongeq b$.  This method preserves soundness and
completeness.

\begin{lemma}(soundness of modified procedure)
  The modified procedure is sound.
\end{lemma}
\begin{proof}
  Follows directly from Lemma~\ref{lem:soundness} since it does not rely on
  the addition of $\select{a}{i}$ for every $\store{a}{i}{\cdot}$.
\end{proof}

For the completeness lemma we take into account the fact that the element
theory is stably infinite.  Thus, if $\select{a}{i}$ is not present we use a
fresh element in the value domain.

\begin{restatable}[completeness of modified
    procedure]{lemma}{modifiedcompleteness}\label{lem:modifiedcompleteness}
  The modified procedure is complete.
\end{restatable}

\begin{longversion}
\begin{proof}
  We redefine the generation of the model $\M_A$ in the proof of
  Lemma~\ref{lem:completeness}.  

  Let $\Weakeqi{a} := \{b\ |\ b\weakeqi a\}$
  denote the set of all array terms that are weakly equivalent modulo $i$ to
  $a$.  We construct $\M_A$ in the following way.
  \[
  \M_A(a)(\mbox{\j}):=\begin{cases}
  \M(\select{b}{i})&\text{if } \select{b}{i}\text{ occurs in the input, } \M(i)=\mbox{\j}\text{ and }
  a\weakeqi b\\
  \fresh{\Weakeqi{a}}&\text{if }\M(i)=\mbox{\j}, i\in\Stores{\Weakeq{a}}\text{ and there are
    no $b$ and $j$}\\
  &\quad\text{such that }\select{b}{j}\text{ occurs in the input,
  }\M(j)=\mbox{\j}\text{ and }a\weakeqi[j] b\\
  \Second_\tau&\text{if }\mbox{\j}=\fresh{\Weakeq{a}}\\
  \First_\tau&\text{otherwise}
  \end{cases}
  \]
  Note that we only change the model for array $a$ and index $i$ if
  $\store{a}{i}{\cdot}$ exists in the input, but $\select{a}{i}$ does not.
  Hence, for definition of the form $v=\select{a}{i}$ and $b =
  \store{a}{i}{v}$ the old proof can be reused.  Thus, we only need to
  show extensionality.

  If $\M(a)=\M(b)$ for two array terms $a$ and $b$, then $a\strongeq
  b$.  Thus, for all indices $i$, $\M_A(a)(i)=\M_A(b)(i)$ holds by
  construction.  Otherwise we distinguish two cases.  If $a$ and $b$
  are not connected in the weak equivalence graph, $\M_A(a)$,
  $\M_A(b)$ differ, e.\,g., on index $\fresh{\Weakeq{a}}$.  If the
  arrays are connected ($a\weakpath{P} b$), there is an
  $i\in\Stores{P}$ such that $a \not\weakcongi b$.  If there is no
  $\select{a'}{j}$ with $a\weakeqi a'$ and $j \strongeq i$, then
  $\M_A(a)(\M(i)) = \fresh{\Weakeqi{a}}$ differs from
  $\M_A(b)(\M(i))$.  Similarly for $b$, if there is no $\select{b'}{k}$ with
  $b\weakeqi b'$ and $k\strongeq i$.
  Otherwise $\M_A(a)(\M(i)) = \M(\select{a'}{j})$ and
  $\M_A(b)(\M(i)) = \M(\select{b'}{k})$ and these values differ since 
  $a\not\weakcongi b$.
\end{proof}
\end{longversion}
\begin{shortversion}
\noindent The proof of this lemma can be found in the extended version of this
paper~\cite{longversion}.
\end{shortversion}

This optimisation enables us to limit the number of additional terms in the
input.  Since we only need to generate (\ref{rule:read-over-weakeq}) lemmas if
the select terms in the conclusion are present after preprocessing, this
optimisation also reduces the number of lemmas.  Furthermore, it is widely
applicable. In fact, the non-bitvector logics defined in the
SMTLIB~\cite{BarST-SMT-10} only allow array sorts where the element theory is
stably infinite.  Thus, only the terms corresponding to instantiations of
Axiom~(\ref{ax:idx}) are required.  In an actual implementation even these
terms could be omitted (see~\cite{DBLP:journals/jsat/BrummayerB09}).

\section{Implementation and Evaluation}\label{sec:implementation}
We implemented the decision procedure described in this paper in our SMT
solver SMTInterpol~\cite{DBLP:conf/spin/ChristHN12}.
Besides the aforementioned preprocessing step that applies (\ref{ax:idx}) to
every \store{\cdot}{\cdot}{\cdot} in the input, we also simplify \ta-formulae by
applying (\ref{ax:read-over-write}) if the index of the store and the index of
the select are syntactically equal.  Furthermore, we contract terms of the
form $\store{\store{a}{i}{v_2}}{i}{v_1}$ to $\store{a}{i}{v_1}$.  We only add
\select{a}{i} to the set of terms contained in the formula if we have
\store{a}{i}{v} in the input and the domain of $v$ is finite.

We represent the weak equivalence relation and the weak equivalence
modulo~$i$ relations in a forest structure, similarly to the
representation of equivalence graph in congruence
solvers~\cite{DBLP:conf/rta/NieuwenhuisO05}.  Every node has an
outgoing edge, and these edges build a spanning tree for every
equivalence class.  The edges point from a child node to the parent
node.  The root node of every tree has no outgoing edge and is the
representative of its equivalence class.

We have to distinguish between strong equivalence, weak equivalence,
and weak equivalence modulo~$i$.  The strong equivalence classes are
already handled by the equality solver.  In our implementation of the
array solver we treat them as indivisible and create a single node for
every strong equivalence class.  To represent the weak equivalence
relations the nodes have up to two outgoing edges, a primary $p$ and a
secondary $s$, see Figure~\ref{alg:data-structures}.  The edges come
from a store operation and correspond to the edges
$\stackrel{i}{\leftrightarrow}$ in the weak equivalence graph.  The
index of the primary edge is stored in the $pi$ field.
\begin{figure}[htbp]
\begin{minipage}{.5\textwidth}
\begin{tabbing}
\textbf{struct} $\type{node}$\\
\qquad\=
       $p : \type{node}$\\
\>     $pi : \type{index}$\\
\>     $s : \type{node}$
\end{tabbing}
\begin{codebox}
\Procname{$\proc{get-rep}(n : \type{node})$}
\zi     \If $n.p = \const{nil}$ \Then $n$
\zi     \Else $\proc{get-rep}(n.p)$
        \End
\end{codebox}
\begin{codebox}
\Procname{$\proc{make-rep}(n : \type{node})$}
\zi  \If $n.p \neq \const{nil}$ \Then
\zi     $\proc{make-rep}(n.p)$
\zi	$\left.\begin{array}{@{}l@{}}
        n.p.p  := n \\
        n.p.pi := n.pi\\
        n.p := \const{nil}
        \end{array}\right\}\mbox{\parbox{2.5cm}{\raggedright invert primary edge}}$
\zi	$\proc{make-rep${}_i$}(n)$
     \End
\end{codebox}
\end{minipage}%
\begin{minipage}{.5\textwidth}
\begin{codebox}
\Procname{$\proc{get-rep${}_i$}(n : \type{node}, i : \type{index})$}
\zi     \If $n.p = \const{nil}$ \Then $n$
\zi     \ElseIf $n.pi \neq i$ \Then $\proc{get-rep${}_i$}(n.p, i)$
\zi     \ElseIf $n.s = \const{nil}$ \Then n
\zi     \Else $\proc{get-rep${}_i$}(n.s, i)$
        \End
\end{codebox}
\begin{codebox}
\Procname{$\proc{make-rep${}_i$}(n : \type{node})$}
\zi   \If $n.s \neq \const{nil}$ \Then
\zi     \If $n.s.pi \neq n.pi$ \Then 
\zi        $n.s := n.s.p$\quad\mbox{\parbox{2.2cm}{move towards representative}}
\zi        $\proc{make-rep${}_i$}(n)$
\zi     \Else
\zi        $\proc{make-rep${}_i$}(n.s)$
\zi	$\left.\begin{array}{@{}l@{}}
        n.s.s := n.s\\
        n.s := \const{nil}\\
        \end{array}\right\}\mbox{\parbox{2.5cm}{\raggedright invert secondary edge}}$
        \End
      \End
\end{codebox}
\end{minipage}%
\caption{Data structure and functions to represent weak equivalence
  relations.  A \type{node} structure is created for every strong
  equivalence class on arrays. It contains two outgoing edges $p,s$
  pointing towards the representative of the weak equivalence classes.
  The functions \proc{get-rep} and \proc{get-rep${}_i$} are used to
  find the representative of the weak equivalence (resp. weak
  equivalence modulo~$i$) class.  The functions $\proc{make-rep}$ and
  $\proc{make-rep${}_i$}$ invert the edges to make a node the
  representative of its weak equivalence classes.\label{alg:data-structures}}
\end{figure}
The primary edge points towards the representative of the weak
equivalence class.  Every primary edge $p$ connects the node
representing (the strong equivalence class of) a store $\store ajv$
with the node representing $a$ and the corresponding index in the $pi$
field is $j$.  Note, however, that the direction of the edge can be
arbitrary, as we invert the edges during the execution of the
algorithm.  If the primary edge is missing the node is the
representative of its weak equivalence class and of all its weak
equivalence modulo~$i$ classes.

While the primary edge is enough to represent the weak equivalence
relation we need another edge to represent weak equivalence modulo
$i$.  The representative of weak equivalence modulo~$i$ is
also found by following the primary edges. However, if the store of
the primary edge is on the index $i$, the secondary edge is followed
instead.  If the secondary edge is missing the node is the
representative of its weak equivalence modulo~$i$ class.  

The equivalence classes are represented as follows. Two
arrays $a$ and $b$ are weakly equivalent iff 
$\proc{get-rep}(a) = \proc{get-rep}(b)$ and
$a\weakeqi b$ iff $\proc{get-rep${}_i$}(a, i) =
\proc{get-rep${}_i$}(b,i)$.
\begin{figure}[htbp]
\begin{minipage}{.5\textwidth}
\begin{codebox}
\Procname{$\proc{add-secondary}(S : \type{index set}, a, b : \type{node})$}
\zi     \If $a = b$ \Then
\zi        \Return
        \End
\zi     \If $a.pi \notin S
            \land \proc{get-rep${}_i$}(a, a.pi) \neq b$ \Then
\zi	   $\proc{make-rep${}_i$}(a)$
\zi        $a.s := b$
        \End
\zi     $\proc{add-secondary}(S \cup \{a.pi\}, a.p, b)$
\end{codebox}
\end{minipage}%
\begin{minipage}{.5\textwidth}
\begin{codebox}
\Procname{$\proc{add-store}(a,b : \type{node}, i: \type{index})$}
\zi     \proc{make-rep}(b)
\zi     \If $\proc{get-rep}(a) = b$ \Then
\zi        $\proc{add-secondary}(\{i\}, a, b)$
\zi     \Else
\zi        $b.p := a$
\zi        $b.pi := i$
        \End
\end{codebox}
\end{minipage}
\caption{The algorithm \proc{add-store} adds a new store edge to the
  data structure updating the weak equivalence classes.  In the else
  case a new primary edge is added to merge two disjoint weak
  equivalence classes.  Otherwise, \proc{add-secondary} inserts new
  secondary edges to merge the necessary weak equivalence modulo~$i$
  classes. \label{alg:add-store}}
\end{figure}

The algorithm proceeds by inserting the store edges one by one,
similarly to the algorithm presented
in~\cite{DBLP:conf/rta/NieuwenhuisO05}.  The algorithm that inserts a
store edge is given in Figure~\ref{alg:add-store}.  The algorithm
first inverts the outgoing edges of one node to make it the
representative of its weak equivalence class.  If the other side of
the store edge lies in a different weak equivalence classes, the store
can be inserted as a new primary edge.

If the nodes are already weakly equivalent the
procedure~\proc{add-secondary} is called.  This procedure follows the
path from the other array $a$ to the array $b$ that was made the
representative.  For every node on this path it checks if a secondary
edge needs to be added.  If the primary edge of the node is labelled
with a store on~$i$, the algorithm first checks if the node is weakly
equivalent modulo~$i$ with $b$ due to the new store edge.
This is the case if no store on $i$ occurred on the path so far and the
new store is also on an index different from $i$.  We use the set $S$
to collect these forbidden indices.  Then if $b$ is not already the 
representative of the weak equivalence modulo~$i$ class, the outgoing 
secondary edges are reversed and a new secondary edge is added.

The complexity of the procedure \proc{add-store} is worst case
quadratic in the size of the weak equivalence class.  This stems from
\proc{make-rep${}_i$} being linear in the size and being called a
linear number of times.  The overall complexity is cubic in the number
of stores in the input formula.  The space requirement, however, is
only linear.  In our current implementation in SMTInterpol this
procedure was not a bottleneck so far.  In SMTInterpol we also keep
the stores that created the primary and secondary edge in the
data-structure.  This allows for computing the paths needed for lemma
generation in linear time.

\begin{example}\label{ex:addedge}
  Figure~\ref{fig:addedge} shows an example of the data structure
  where the primary edges are labelled by the index of the
  corresponding store.  This data structure represents only one weak
  equivalence class with the representative node 0.  The resulting
  data structure after adding a store with index $k$ between nodes 0
  and 4 is shown on the right.  Since nodes 0 and 4 were
  already in the same weak equivalence class, secondary edges were
  added.  \pgfdeclarelayer{back} \pgfsetlayers{back,main}
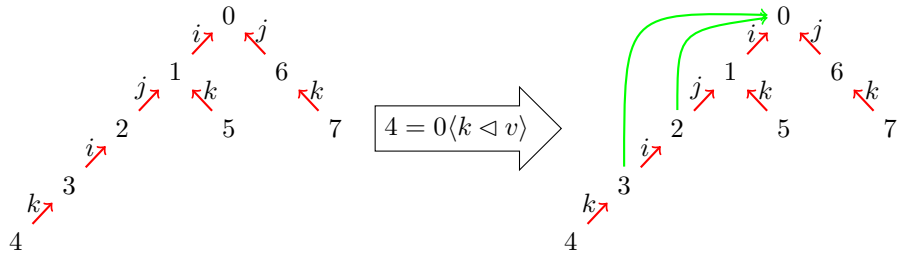
\begin{figure}[htbp]
\hfill
\begin{tikzpicture}[edge from parent/.style={draw,red,<-,thick},level
    distance=0.75cm,sibling distance=1.4cm,index/.style={black,inner sep=0cm}]
  \begin{scope}
  \node[name=n0] {0}
  child[thick] {
    node[name=n1] {1}
    child {
      node[name=n2] {2}
      child {
        node[name=n3] {3}
        child {
          node[name=n4] {4}
          edge from parent node[above left,index] {$k$}
        }
        child[missing] {
          node {}
        }
        edge from parent node[above left,index] {$i$}
      }
      child[missing] {
        node {}
      }
      edge from parent node[above left,index] {$j$}
    }
    child {
      node[name=n5] {5}
      edge from parent node[above right,index] {$k$}
    }
    edge from parent node[above left,index] {$i$}
  }
  child[thick] {
    node[name=n6] {6}
    child[missing] {
      node {}
    }
    child {
      node[name=n7] {7}
      edge from parent node[above right,index] {$k$}
    }
    edge from parent node[above right,index] {$j$}
  };
  \end{scope}
  \node[draw,shape=single arrow] at(3,-1.5) {$4 = \store{0}{k}{v}$};
  \begin{scope}[xshift=7.3cm]
  \node[name=n0] {0}
  child[thick] {
    node[name=n1] {1}
    child {
      node[name=n2] {2}
      child {
        node[name=n3] {3}
        child {
          node[name=n4] {4}
          edge from parent node[above left,index] {$k$}
        }
        child[missing] {
          node {}
        }
        edge from parent node[above left,index] {$i$}
      }
      child[missing] {
        node {}
      }
      edge from parent node[above left,index] {$j$}
    }
    child {
      node[name=n5] {5}
      edge from parent node[above right,index] {$k$}
    }
    edge from parent node[above left,index] {$i$}
  }
  child[thick] {
    node[name=n6] {6}
    child[missing] {
      node {}
    }
    child {
      node[name=n7] {7}
      edge from parent node[above right,index] {$k$}
    }
    edge from parent node[above right,index] {$j$}
  };
  \draw[green,thick,->] (n3) .. controls +(0,2.3) .. (n0);
  \draw[green,thick,->] (n2) .. controls +(0,1.3) .. (n0);
  \end{scope}
\end{tikzpicture}
\hfill\hbox{}
\caption{\label{fig:addedge} Weak equivalence classes represented by a
graph using primary and secondary edges.  The short direct edges are
primary edges, the long bended edges are secondary edges.  Each
primary edge represents a store edge between the connected nodes and is
labelled by the index of the store.  The secondary edges in the right graph
were created by a store edge on index $k$ between node 0 and 4 as described
in Example~\ref{ex:addedge}.
}
\end{figure}

These secondary edges are needed to connect the weak equivalence
modulo $i$ and modulo $j$ classes.  Figure~\ref{fig:mergei} shows how
the first secondary edge connects the two weak equivalence modulo $i$
classes rooted at nodes 0 resp.\ 3.  This is necessary since there is
now a new path using the edge from 4 to 0.  Note that no secondary
edge is added to node~1, since nodes 1, 2, and 5 are still not weakly
equivalent modulo $i$ to the other nodes.
Figure~\ref{fig:mergej} shows the connection between the two weak equivalence
modulo $j$ classes rooted at nodes 0 resp.\ 2.  The weak equivalence modulo
$j$ class rooted at node 6 is not affected by a new edge between nodes 0 and 4
since these nodes are on a different path.
\begin{figure}[htbp]
\subfigure[Merging weak equivalence modulo $i$ classes.\label{fig:mergei}]{
\begin{tikzpicture}[edge from parent/.style={draw,red,<-,thick},level
    distance=0.75cm,sibling distance=1.4cm,index/.style={black,inner sep=0cm}]
  \node[name=n0] {0}
  child[thick] {
    node[name=n1] {1}
    child {
      node[name=n2] {2}
      child {
        node[name=n3] {3}
        child {
          node[name=n4] {4}
          edge from parent node[above left,index] {$k$}
        }
        child[missing] {
          node {}
        }
        edge from parent node[above left,index] {$i$}
      }
      child[missing] {
        node {}
      }
      edge from parent node[above left,index] {$j$}
    }
    child {
      node[name=n5] {5}
      edge from parent node[above right,index] {$k$}
    }
    edge from parent node[above left,index] {$i$}
  }
  child[thick] {
    node[name=n6] {6}
    child[missing] {
      node {}
    }
    child {
      node[name=n7] {7}
      edge from parent node[above right,index] {$k$}
    }
    edge from parent node[above right,index] {$j$}
  };
  \draw[green,thick,->] (n3) .. controls +(0,2.3) .. (n0);
  \draw[black,dashed,thick] (n4) .. controls +(0,3.5) .. node[left]{$k$} (n0);
  \begin{pgfonlayer}{back}
  \draw[rotate=45,fill=white!80!black] (n6) ellipse [x radius=.5cm,y radius=1.7cm];
  \draw[rotate=-45,fill=white!80!black] ($(n3)!.5!(n4)$) ellipse [x radius=.5cm,y radius=1cm];
  \draw[fill=white!90!black] ($(n1)-(0,.6)$) circle [radius=.9cm];
  \end{pgfonlayer}
\end{tikzpicture}
}
\hfill
\subfigure[Merging weak equivalence modulo $j$ classes.\label{fig:mergej}]{
\begin{tikzpicture}[edge from parent/.style={draw,red,<-,thick},level
    distance=0.75cm,sibling distance=1.4cm,index/.style={black,inner sep=0cm}]
  \node[name=n0] {0}
  child[thick] {
    node[name=n1] {1}
    child {
      node[name=n2] {2}
      child {
        node[name=n3] {3}
        child {
          node[name=n4] {4}
          edge from parent node[above left,index] {$k$}
        }
        child[missing] {
          node {}
        }
        edge from parent node[above left,index] {$i$}
      }
      child[missing] {
        node {}
      }
      edge from parent node[above left,index] {$j$}
    }
    child {
      node[name=n5] {5}
      edge from parent node[above right,index] {$k$}
    }
    edge from parent node[above left,index] {$i$}
  }
  child[thick] {
    node[name=n6] {6}
    child[missing] {
      node {}
    }
    child {
      node[name=n7] {7}
      edge from parent node[above right,index] {$k$}
    }
    edge from parent node[above right,index] {$j$}
  };
  \draw[green,thick,->] (n2) .. controls +(0,1.3) .. (n0);
  \draw[black,dashed,thick] (n4) .. controls +(0,3.5) .. node[left]{$k$} (n0);
  \begin{pgfonlayer}{back}
  \draw[rotate=-45,fill=white!80!black] (n3) ellipse [x radius=.5cm,y radius=1.6cm];
  \draw[rotate=45,fill=white!90!black] ($(n6)!.5!(n7)$) ellipse [x radius=.5cm,y radius=.9cm];
  \draw[rounded corners,fill=white!80!black] (n0.north west) -- ($(n1.west)-(.2,0)$) --
  ($(n5.south east)-(0,.2)$) -- (n0.north east) -- cycle;
  \end{pgfonlayer}
\end{tikzpicture}
}
\caption{\label{fig:mergeweakeqi}Secondary edges merge weak equivalence modulo
  $i$ classes.}
\end{figure}
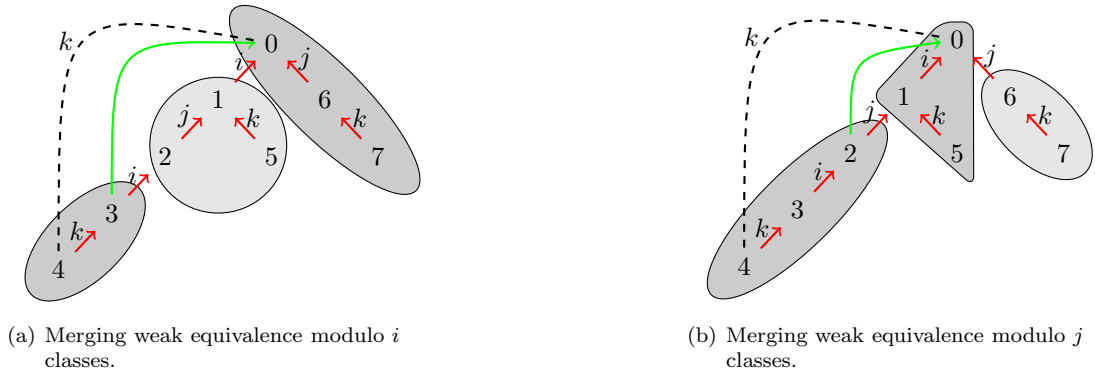
\end{example}

We implemented this decision procedure in our SMT solver
SMTInterpol~\cite{DBLP:conf/spin/ChristHN12} and tested it on the benchmarks
from the QF\_AX and QF\_AUFLIA devisions of the SMTEVAL 2013 benchmarks.  We
solved all benchmarks in 1:32 resp.\ 10:45 minutes without running into a
timeout of 10 minutes.  According to the data from the SMTEVAL, no other
solver was able to solve all benchmarks in these divisions.  We defer an
up-to-date comparison to the SMTCOMP 2014.

\section{Conclusion and Future Work}
We presented a new decision procedure for the extensional theory of arrays.
This procedure exploits weak equalities to limit the number of axiom
instantiations.  The instantiations produced by the decision procedure
presented in this paper can be restricted to terms already present in the
input formula.  Furthermore we discussed an implementation based on a graph
structure similar to congruence closure graphs.  This decision procedure is
implemented in our SMT solver \si~\cite{DBLP:conf/spin/ChristHN12}.  We plan
to implement a variant of the quantifier-free interpolation for
arrays~\cite{DBLP:journals/corr/abs-1204-2386} based on the lemmas generated
by this decision procedure.  Since these lemmas only generate mixed
equalities, proof tree preserving
interpolation~\cite{DBLP:conf/tacas/ChristHN13} can be used.

\bibliographystyle{splncs03}
\bibliography{paper}

\end{document}